\documentclass[prd,tightenlines,nofootinbib,superscriptaddress]{revtex4}

\usepackage{amsfonts,amssymb,amsthm,bbm}

\usepackage{amsmath}

\usepackage{hyperref}

\usepackage{color,psfrag}
\usepackage[dvips]{graphicx}

\newcommand{\C}{{\mathbb C}}
\newcommand{\N}{{\mathbb N}}
\newcommand{\R}{{\mathbb R}}
\newcommand{\Z}{{\mathbb Z}}

\newcommand{\cE}{{\mathcal E}}

\newcommand{\cJ}{{\mathcal J}}

\newcommand{\cH}{{\mathcal H}}
\newcommand{\cM}{{\mathcal M}}

\newcommand{\cP}{{\mathcal P}}

\newcommand{\cS}{{\mathcal S}}

\newcommand{\SU}{\mathrm{SU}}

\newcommand{\SL}{\mathrm{SL}}
\newcommand{\GL}{\mathrm{GL}}
\newcommand{\SO}{\mathrm{SO}}
\renewcommand{\O}{\mathrm{O}}
\newcommand{\U}{\mathrm{U}}

\newcommand{\be}{\begin{equation}}
\newcommand{\ee}{\end{equation}}
\newcommand{\beq}{\begin{eqnarray}}
\newcommand{\eeq}{\end{eqnarray}}
\newcommand{\bes}{\begin{eqnarray}}
\newcommand{\ees}{\end{eqnarray}}

\newcommand{\mat} [2] {\left ( \begin{array}{#1}#2\end{array} \right ) }

\renewcommand{\u}{{\mathfrak{u}}}
\newcommand{\su}{{\mathfrak{su}}}

\renewcommand{\sl}{{\mathfrak{sl}}}
\newcommand{\gl}{{\mathfrak{gl}}}
\newcommand{\so}{{\mathfrak{so}}}

\def\cc{{\mathfrak{C}}}

\newcommand{\la}{\langle}
\newcommand{\ra}{\rangle}

\newcommand{\tr}{{\mathrm{Tr}}}
\newcommand{\f}{\frac}

\def\nn{\nonumber}

\def\eps{\epsilon}

\newcommand{\id}{\mathbb{I}}

\def\vcJ{\vec{\cJ}}
\def\vJ{\vec{J}}
\def\vK{\vec{K}}
\def\vL{\vec{L}}

\def\vX{\vec{X}}

\def\bz{\bar{z}}

\def\bw{\bar{w}}

\def\bt{\bar{t}}
\def\bu{\bar{u}}

\def\hn{\hat{n}}

\def\hJ{\hat{J}}
\def\hK{\hat{K}}
\def\hE{\hat{E}}
\def\hcE{\hat{\cE}}
\def\hC{\hat{\mathfrak{C}}}

\def\tj{\tilde{j}}
\def\tk{\tilde{k}}
\def\vtau{\vec{\tau}}
\def\veta{\vec{\eta}}

\def\tM{\widetilde{M}}
\def\bM{\overline{M}}
\def\tz{\tilde{z}}
\def\tG{\widetilde{G}}

\def\halpha{\hat{\alpha}}
\def\hbeta{\hat{\beta}}
\def\hDelta{\widehat{\Delta}}
\def\ad{a^{\dagger}}
\def\bd{b^{\dagger}}

\newtheorem{theorem}{Theorem}[section]

\newtheorem{prop}[theorem]{Proposition}

\begin{document}

\title{Deformations of Lorentzian Polyhedra: \vspace*{1mm}\\Kapovich-Millson phase space and SU(1,1) Intertwiners}

\author{{\bf Etera R. Livine}}\email{etera.livine@ens-lyon.fr}
\affiliation{Univ Lyon, Ens de Lyon, Universit\'e Claude Bernard, CNRS,
Laboratoire de Physique, F-69342 Lyon, France}

\date{\today}

\begin{abstract}

We describe the Lorentzian version of the Kapovitch-Millson phase space for polyhedra with $N$ faces. Starting with the Schwinger representation of the $\su(1,1)$ Lie algebra in terms of a pair of complex variables (or spinor), we define the phase space for a space-like vectors in the three-dimensional Minkowski space $\R^{1,2}$. Considering $N$ copies of this space, quotiented by a closure constraint forcing the sum of those 3-vectors to vanish, we obtain the phase space for Lorentzian polyhedra with $N$ faces whose normal vectors are  space-like, up to Lorentz transformations. We identify a generating set of $\SU(1,1)$-invariant observables, whose flow by the Poisson bracket generate both area-preserving and area-changing deformations. We further show that the area-preserving observables form a $\gl_{N}(\R)$ Lie algebra and that they generate a $\GL_{N}(\R)$ action on Lorentzian polyhedra at fixed total area. That action is cyclic and all Lorentzian polyhedra can be obtained from a totally squashed polyhedron (with only two non-trivial faces) by a $\GL_{N}(\R)$ transformation. All those features carry on to the quantum level, where quantum Lorentzian polyhedra are defined as $\SU(1,1)$ intertwiners between unitary $\SU(1,1)$-representations from the principal continuous series.

Those $\SU(1,1)$-intertwiners are the building blocks of spin network states in loop quantum gravity in 3+1 dimensions for time-like slicing and  the present  analysis applies to deformations of the quantum geometry of time-like boundaries  in quantum gravity, which is especially relevant to the study of quasi-local observables and holographic duality.

\end{abstract}

\maketitle

\section{Introduction}

The question of quantum gravity is a great physical motivation to explore the fundamental structures of geometry. From a conservative perspective, its goal can be understood as defining a intrinsically discrete notion of geometry, due to the introduction of the Planck length, while still carrying an action of the continuous group of diffeomoprhisms. This would achieve the quantization of geometry. Following this line of research, the loop quantum gravity framework proposes quantum states of 3d geometry and aims at describing their evolution thereby generating the 4d space-time (see \cite{Rovelli:2011eq,Perez:2004hj,Thiemann:2007zz} for reviews).  For space-like 3d hypersurfaces, those spin network states are graphs dressed with algebraic data from the representation theory of the Lie group $\SU(2)$. These can be interpreted as discrete geometries, named ``twisted geometries'' generalizing Regge triangulations \cite{Freidel:2010aq,Dupuis:2012yw}. Their fundamental building blocks are $\SU(2)$ intertwiners, i.e. $\SU(2)$-invariant states in the tensor product of $\SU(2)$ representations, that are understood as the quantum counterpart of 3d polyhedra \cite{Freidel:2009ck,Bianchi:2010gc,Livine:2013tsa}. These quantum polyhedra are then glued together to form a discrete quantum 3d geometry.
The purpose of the present paper is to investigate the  extension of the standard framework to time-like hypersurfaces, with spin network states made from $\SU(1,1)$ intertwiners representing quantized Lorentzian polyhedra. This is directly applicable to loop quantum gravity with time-like slicing as introduced and studied in \cite{Alexandrov:2005ar,Conrady:2010kc,Conrady:2010vx,Conrady:2010sx,Liu:2017bfk}, but is more generally relevant to the issue of defining time-like boundary at the quantum level in quantum gravity. Such boundaries are necessary as soon as we study quasi-local observables or investigate asymptotic boundary conditions for instance in the context of holographic dualities.
Here we will define the classical phase space of Lorentzian polyhedra in $\R^{1,2}$ with space-like normals to their faces, which leads back to the Hilbert space of $\SU(1,1)$ intertwiners after quantization. And we will describe and analyze deformations of Lorentzian polyhedra, thereby leading to a deeper understanding of the structure of the space of quantum Lorentzian polyhedra.

More precisely, we will tackle the extension to the Lorentzian signature of the Kapovitch-Millson phase space for flat polyhedra in $\R^{1,2}$. Let us quickly review the definitions and results derived in the Euclidean case. By Minkowski theorem, a flat convex polyhedra in $\R^3$ is uniquely determined by the normal vectors to its faces, with the sole constraint that these normal vectors sum to 0. Considering polyhedra with $N$ faces, we look at the space of polyhedra with fixed face areas $\{a_{i}\}_{i=1..N}\in(\R_{+}^*)^{\times N}$ up to 3d rotations:
\be
P_{\{a_{i}\}_{i=1..N}}=\Big{\{}
\vX^{i}=a_{i}\,\hn^{i}\,\in\R^3\,,\,\, |\hn^{i}|=1
\Big{\}}//\Big{(}\sum_{i=1}^N\vX^{i}=0\Big{)}
\quad
\sim
(\cS_{2})^{\times N}//\SO(3)
\,.
\ee
The 2-spheres, as complex manifolds, are provided with the usual symplectic structure rescaled by the fixed face areas:
\be
\forall i\,,\quad
\{X_{a}^i,X_{b}^i\}=\eps_{abc}X_{c}^i
\,,\quad
\{\hn_{a}^i,\hn_{b}^i\}=\f1{a_{i}}\eps_{abc}\hn_{c}^i
\,,
\ee
where $\eps_{abc}$ is the rank-3 completely antisymmetric tensor. Then we take the symplectic quotient by the closure constraint, $\sum_{i=1}^N\vX^{i}=0$, which generates simultaneous 3d rotations on all the $N$ vectors $\vX^i$. This defines the Kapovitch-Millson phase space for flat polyhedra in $\R^3$ \cite{kapovich1996}. Its dimension is $2N-6$, where the 6 accounts for the closure constraint and the invariance under $\SO(3)$.

If we would like to unfreeze the face areas $a_{i}$, we now need to enlarge the phase space by further introducing its conjugate variable. The simpler way to proceed is to embed $\R^3$ in $\C^2$ and introduce spinors, or pairs of complex variables, $z_{i}\in\C^2$, and define a larger space of framed polyhedra up to 3d rotations \cite{Livine:2013tsa}:
\be
\cP=\Big{\{}
z_{i}\,\in\C^2
\Big{\}}
//
\Big{(}\sum_{i=1}^N|z_{i}\ra\la z_{i}|=\sum_{i=1}^N\la z_{i}|z_{i}\ra\id\Big{)}
\quad
\sim
\C^{2N}//\SU(2)
\,,
\ee
where we used the bra-ket notations to define the closure constraint:
\be
|z\ra=\mat{c}{z^0\\z^1}\in\C^2\,,\qquad
\la z |=\mat{cc}{\bz^0 & \bz^1}
\,.
\ee
The complex variables are endowed with the canonical Poisson bracket, $\{z^A_{i},\bz^B_{j}\}=-i\delta_{ij}\delta^{AB}$, and one recovers the 3-vectors by projecting the spinors onto the Pauli matrices $\sigma_{a}$:
\be
X^i_{a}=\la z_{i}|\sigma_{a}|z_{i}\ra
\,,\quad
|\vX^i|=\la z_{i}|z_{i}\ra
\,,
\ee
so that the closure constraint defined for the complex variables exactly reproduces the original closure constraint, $\sum_{i}\vX^i=0$. The face areas, given by the vectors' norms, are defined by the norm squared of the spinors. Moreover, the extension of $\R^3$ to $\C^2$ leads to one extra variable per face, given by the phase of the spinor and interpreted as an angle, or 2d frame, attached to each face \cite{Freidel:2009ck}. This is the reason why these structures are referred to as ``framed polyhedra''. In the context of twisted geometries, this extra angle plays an important role and is used to encode the extrinsic curvature of the 3d hypersurface in the 4d space-time \cite{Freidel:2010aq,Freidel:2010bw}.
This actually was  the starting point for the spinorial reformulation of loop quantum gravity \cite{Borja:2010rc,Livine:2011gp,Bonzom:2012bn,Livine:2013wmq,Alesci:2015yqa,Alesci:2016dqx}, where the spinors become the label of coherent spin network states and the quantum gravity constraints are written as differential operators in those complex variables.
In the context of the spinfoam framework providing a quantized path integral for discretized gravity (which can be interpreted as a history formulation of loop quantum gravity), these coherent spin network techniques have also been used to derive the semi-classical behavior of spinfoam amplitudes for large spin (i.e. the physics of quanta of geometry at mesoscopic scales, large compared to the Planck scale but still very small compared to our scale) \cite{Livine:2007vk,Livine:2007ya,Freidel:2007py,Barrett:2009gg,Barrett:2009mw,Barrett:2010ex}.

An important feature of this formalism is the resulting action of the unitary group $\U(N)$ on polyhedra:
\be
z_{i}\in\C^{2N}\,\longmapsto\,
\tz_{i}=\sum_{j}U_{ij}z_{j}\,\qquad
U\in\U(N).
\ee
This action commutes with the closure constraint and leaves the area invariant \cite{Freidel:2010tt}. The $\U(N)$-action is moreover cyclic on the space of framed polyhedra with fixed total area (defined as the sum of the faces' areas). It allows to explore the whole space of polyhedra with arbitrary face shape and area while keeping the overall boundary area fixed by acting with $\U(N)$ transformations on a totally squashed configuration $\{Z_{i}\}_{i=1..N}$ where the polyhedron only has two non-trivial faces,
\be
Z_{1}=\mat{c}{Z_{1}^0\\Z_{1}^1}
\,,\quad
Z_{2}=\mat{c}{-\bar{Z}_{1}^1\\\bar{Z}_{1}^0}
\,,\quad
Z_{k\ge3}=0
\,.
\ee
These $\U(N)$ transformations are generated  by the $\SU(2)$-invariant scalar products $E_{ij}\equiv \la z_{i}|z_{j}\ra$ ,which Poisson-commute with the total area.  Furthermore, one can identify a generating set of $\SU(2)$-invariant observables, complementing the $E_{ij}$'s with other observables which do not Poisson-commute with the total area. This leads to $\SO^{*}(2N)$ transformations, which not only describe area-preserving deformations of the polyhedra (given by the $\U(N)$ subgroup) but more general area-changing deformations \cite{Girelli:2017dbk}.

Upon quantization, the phase space of framed polyhedra leads to the Hilbert space of $\SU(2)$ intertwiners, interpreted as quantum polyhedra. The $\U(N)$ and $\SO^*(2N)$ actions are preserved at the quantum level and lead to the $\U(N)$ formalism for $\SU(2)$ Intertwiners \cite{Girelli:2005ii,Freidel:2009ck,Freidel:2010tt,Livine:2013tsa,Girelli:2017dbk}. Finally, one can define coherent intertwiner states, labeled by the spinors and peaked on classical polyhedra \cite{Freidel:2010tt,Girelli:2017dbk}.

\bigskip

The goal of the present work is to extend all these structures, definitions and results, to the Lorentzian case. A study of the classical phase space for $\SL(2,\C)$ intertwiners had been done in \cite{Dupuis:2011wy}. Here, we focus on $\SU(1,1)$ instead of $\SL(2,\C)$ and use similar methods to define the phase space of Lorentzian polyhedra and further proceed to its quantization to the Hilbert space of $\SU(1,1)$ intertwiners. Although we mostly focus the geometrical interpretation of the structures as 3d objects (to be embedded in 4d space-time), our analysis can be also considered as complementing the work on $\SU(1,1)$ spin networks done in the context of 2+1d quantum gravity (as 2d objects to be embedded in a 3d space-time) \cite{Girelli:2015ija} and exploring the finer structure of $\SU(1,1)$ intertwiners.

\section{The phase space of space-like vectors in 2+1-dimensions}

\subsection{Schwinger representation of the $\su(1,1)$ algebra}

Let us start with a pair of complex variables $(z,w)\in\C^{2}$ endowed with a canonical Poisson bracket:
\be
\{z,\bz\}=\{w,\bw\}=-i
\,,\qquad
\{z,w\}=\{z,\bw\}=\{\bz,w\}=\{\bz,\bw\}=0
\,.
\ee
This provides a presentation of the $\su(1,1)$ Lie algebra, {\it \`a la Schwinger}, with the generators constructed as quadratic polynomials in those variables:
\be
\label{su11}
J_{3}=\f12(z\bz-w\bw)
\,,\quad
K_{+}=\f12(\bz^{2}-w^{2})
\,,\quad
K_{-}=\overline{K_{+}}=\f12(z^{2}-\bw^{2})
\,,
\ee
with the expected Poisson brackets:
\be
\{J_{3},K_{\pm}\}=\mp iK_{\pm}
\,,\quad
\{K_{+},K_{-}\}=2iJ_{3}
\,.
\ee
We can switch to the usual real basis of the $\su(1,1)$ Lie algebra\footnotemark:
\be
K_{\pm}=K_{1}\pm i K_{2}
\,,\qquad
\{K_{1},K_{2}\}=-J_{3}
\,,\quad
\{J_{3},K_{1}\}=K_{2}
\,,\quad
\{J_{3},K_{2}\}=-K_{1}
\ee
\footnotetext{The Poisson bracket can be written in terms of a 3d Minkowskian metric $\eta$ and the totally antisymmetric tensor on 3 indices $\eps$:
\be
\{\cJ_{a},\cJ_{b}\}=\epsilon_{ab}{}^{c}\cJ_{c}
\,,\quad
\epsilon_{ab}{}^{c}=\epsilon_{abd}\eta^{cd}
\,,\quad
\eta^{cd}=(-1)^{\delta_{c3}}\delta^{cd}
\,.
\nn
\ee
}
We will use a vectorial notation $\vcJ=(J_{3},K_{1},K_{2})$.
The quadratic Casimir of the $\su(1,1)$ algebra admits a simple expression in terms of the complex variables:
\be
\cc=\vcJ^{2}=J_{3}^{2}-K_{+}K_{-}=J_{3}^{2}-K_{1}^{2}-K_{2}^{2}=-\cE^{2}
\,,\qquad
\cE=\f i2(\bz\bw-zw)\in\R
\,,\quad
\{\cE,\vcJ\}=0
\,,
\ee
so that the ``energy'' $\cE$ gives the Lorentzian norm of the space-like vector $\vcJ$, up to a sign.  Indeed, defining the norm as $|\vcJ|\equiv\sqrt{-\vcJ^{2}}$ for a space-like vector, we have $\cE=\pm |\vcJ|$ with the sign depending on the complex variables $(z,w)$. Switching the sign for one of the two complex variables, say $w\mapsto -w$, changes  the sign of the energy without affecting the 3-vector:
\be
\left|
\begin{array}{rcl}
(z,w)&\mapsto& (z,-w)\\
\vcJ &\mapsto& \vcJ\\
\cE&\mapsto&-\cE
\end{array}
\right.
\ee
There it is always possible to take $\cE\ge 0$ without changing the 3-vector $\vcJ$. Keeping this sign ambiguity in mind, we will nevertheless loosely refer to $\cE$ as the norm of the space-like vector whenever this does not lead to any confusion.

Switching the signs of both complex variables, $(z,w)\mapsto -(z,w)$ leaves both $\vcJ$ and $\cE$ invariant. While exchanging them $(z,w)\mapsto (\bw,\bz)$ induces a total switch of signs in the $\su(1,1)$ generators $(\vcJ,\cE)\mapsto -(\vcJ,\cE)$.

\medskip

This construction provides a Poisson bracket structure for three-dimensional space-like vectors in the 3d Minkowski space $\R^{1,2}$ with signature $(+,-,-)$. This is clearly not a symplectic manifold, since it has an odd dimension. If we want to work on a symplectic manifold, we either work with the 2d space of space-like vectors with fixed norm $\cE=\cE_{0}$, or we work with the original 4d space spanned by the complex pair $(z,w)$. In the latter case, it is possible to identify the missing variable, aside the 3-vector $\vcJ$, in order to parametrize the whole complex space. For that purpose, it is convenient to introduce a  change of variables:
\be
\label{utdef}
u=\f1{\sqrt2}(z-\bw)
\,,\quad
t=\f1{\sqrt2}(z+\bw)
\,,\quad
\{u,t\}=0
\,,\quad
\{u,\bt\}=-i
\,,\quad
\{u,\bu\}=0
\,.
\ee
The $\su(1,1)$ admits a simpler form:
\be
ut=K_{-}
\,,\quad
\bu\bt=K_{+}
\,,\quad
u\bt=J_{3}+i\cE
\,,
\ee
from which we see that the generators $\vcJ$, as well as the norm (or energy) $\cE$, is invariant under real rescaling of the complex variables:
\be
\left|
\begin{array}{lcr}
u&\rightarrow&\varsigma\,e^{\f12\lambda}\,u
\\
t&\rightarrow&\varsigma\,e^{-\f12\lambda}\,t
\end{array}
\right.
\,,\qquad
\lambda\in\R
\,,\,\,
\varsigma=\pm
\,.
\ee
We can promote this to an actual observable on the 4d phase space:
\be
\label{lambda}
\lambda
\equiv
\ln{\left|
\f{u}t
\right|}
=
\ln{\left|
\f{z-\bw}{z+\bw}
\right|}
\,,\qquad
\{\cE,\lambda\}=1
\,,
\ee
which provides a conjugate variable to the norm $\cE$. Switching the sign of both complex variables $(u,t)\mapsto -(u,t)$ obviously leaves $(\vcJ,\lambda)$ invariant. Exchanging the two variables $(u,t)\mapsto (t,u)$ produces a straightforward switch of sign for $\lambda$ (and $\cE$), that is $(\vcJ,\lambda)\mapsto(\vcJ,-\lambda)$. 

The key point is that  the four variables $(\vcJ,\lambda)$ uniquely determine the pair of complex variables $(z,w)\in\C^{2}$, up to signs, except in the degenerate case $\vcJ=0$. This degenerate point corresponds to $t=0$ or $u=0$ (with $\lambda$ being $\pm\infty$), or equivalently to $\bw=\pm z$. More precisely:
\begin{prop}
\label{iso}
Pairs of non-vanishing complex variables $(u,t)\in\C^{2}$, with $u\ne 0$ and $t\ne0$, are in bijection with quadruplets $(\vcJ,\lambda, \eps,\varsigma)\in\R^{1,2}\times\R\times\{\pm\}\times\{\pm\}$ where $\vcJ$ is a non-vanishing space-like 3-vector (possibly null-like), $\vcJ^{2}\le 0$ and $\vcJ\ne 0$. The mapping is given by:
\be
ut=K_{-}
\,,\quad
u\bt=J_{3}+i\eps|\vcJ|
\,,\quad
e^{\lambda}=\f{|u|}{|t|}
\,,\quad
\varsigma=\mathrm{sign}(\mathrm{Arg}(u)-\mathrm{Arg}(t))\,,
\ee
where  the space-like 3-vector $\vcJ$ is $(J_{3},\mathrm{Re}\, K_{-},-\mathrm{Im}\, K_{-})$, $|\vcJ|=\sqrt{-\vcJ^{2}}=\sqrt{|K_{-}|^{2}-J_{3}^{2}}\ge 0$ is its norm and the arguments $\mathrm{Arg}()$ are defined in $[0,2\pi[$. When $\vcJ$ is null-like, its norm vanishes and the sign $\eps$ is irrelevant.

\end{prop}
\begin{proof}
The mapping from $(u,t)$ to $(\vcJ,\lambda, \eps,\varsigma)$ is well-defined. Let us check the inverse mapping.
We first compute the modulus of $u$ and $t$:
\be
\left|\begin{array}{l}
|u||t|=|K_{-}|\\
|u|=|t|\,e^{\lambda}
\end{array}\right.
\quad\Rightarrow\quad
|u|=\sqrt{|K_{-}|}e^{+\f\lambda 2}
\,,\,\,
|t|=\sqrt{|K_{-}|}e^{-\f\lambda 2}
\ee
Then we introduce the phases $\theta,\phi\in[0,2\pi[$:
\be
K_{-}=e^{i\phi}|K_{-}|
\,,\quad 
J_{3}+i\eps|\vcJ|=e^{i\theta}|K_{-}|\,.
\ee
Changing $\eps$ into $-\eps$ would map $\theta$ into $2\pi-\theta$.
This finally allows to get $t$ and $u$:
\be
\f{\bt}{t}=\f{J_{3}+i\eps|\vcJ|}{K_{-}}
\,,\quad
\f{u}{\bu}=\f{J_{3}+i\eps|\vcJ|}{\bar{K}_{-}}
\quad\Rightarrow\quad
t=\varsigma\,|t|\,e^{i\f{\phi-\theta}2}
\,,\quad
u=\varsigma\,|u|\,e^{i\f{\phi+\theta}2}
\,.
\ee
Since the angle $\phi+\theta$ is by construction larger or equal than $\phi-\theta$ in $[0,2\pi[$, the sign $\varsigma$ allows to explore the missing sector.

\end{proof}


\medskip

Finally, we would like to point out that the original $\su(1,1)$ presentation that we defined above in \eqref{su11} admits a simple geometrical interpretation. Indeed, we can see that the $\su(1,1)$ generators actually are the sum of two sets of $\su(1,1)$ generators, one associated to $z$ and another one associated to $w$:
\be
\cJ_{a}=j_{a}(z)+\tj_{a}(w)
\,,\quad
\left|
\begin{array}{lcl}
2j_{3}=z\bz
\\
2k_{+}=\bz^{2}
\\
2k_{-}=z^{2}
\end{array}
\right.
\,,\quad
\left|
\begin{array}{lcl}
2\tj_{3}=-w\bw
\\
2\tk_{+}=-w^{2}
\\
2\tk_{-}=-\bw^{2}
\end{array}
\right.
\,,\quad
\left|
\begin{array}{lcl}
\{j_{3},k_{\pm}\}=\mp i k_{\pm}
\\
\{k_{+},k_{-}\}=2ij_{3}
\end{array}
\right.
\,,\quad
\left|
\begin{array}{lcl}
\{\tj_{3},\tk_{\pm}\}=\mp i \tk_{\pm}
\\
\{\tk_{+},\tk_{-}\}=2i\tj_{3}
\end{array}
\right.
\,.
\ee
These two 3-vectors are null-vectors, with vanishing Casimirs:
\be
j_{3}^{2}-k_{-}k_{+}=\tj_{3}^{2}-\tk_{-}\tk_{+}=0
\,;
\ee
the first one $\vec{j}$ is future-oriented, $j_{3}\ge 0$, while the second one is past-oriented, $\tj_{3}\le 0$, so that their sum produces arbitrary space-like 3-vectors. At the quantum level, this translates into the fact that the tensor product of two null-like unitary $\SU(1,1)$-representations, one with positive weights while the other one with negative weights, decomposes into all possible space-like unitary $\SU(1,1)$-representations.
This is the same doubling trick used in \cite{Dupuis:2011wy} to build an unitary presentation of the $\sl(2,\C)$ algebra.

\subsection{Exponentiating the $\SU(1,1)$ action}

The 3-vector components $\cJ_{a}$ forms a $\su(1,1)$ algebra and generates $\SU(1,1)$ Lorentz transformations. We can compute their Poisson bracket with the complex variables that define infinitesimal $\su(1,1)$ transformations and exponentiate their action into finite $\SU(1,1)$ transformations.
To this purpose, we introduce the Lorentzian Pauli matrices:
\be
\tau_{3}=\mat{cc}{1 & 0 \\ 0 & -1}
\,,\,\,
\tau_{1}=\mat{cc}{ 0 & 1 \\ -1 &0}
\,,\,\,
\tau_{2}=\mat{cc}{ 0 & -i \\-i &0}
\,.
\nn
\ee
These matrices square to the identity, $\tau_{3}^{2}=\id$ but $\tau_{1}^{2}=\tau_{2}^{2}=-\id$, and satisfy the $\su(1,1)$ commutation relations:
\be
[\tau_{3},\tau_{1}]=2i\tau_{2}
\,,\quad
[\tau_{1},\tau_{2}]=-2i\tau_{3}
\,,\quad
[\tau_{2},\tau_{3}]=2i\tau_{1}
\,.
\nn
\ee
We compute the Poisson brackets of $J_{3}$ and  $K_{1,2}$ with $z$ and $w$:
\be
\begin{array}{l}
\left\{\veta\cdot\vcJ,\mat{c}{z\\\bz}\right\}
\,=\,
\f i2\veta\cdot\vtau\,\mat{c}{z\\\bz}
\,,\qquad
e^{\{\veta\cdot\vcJ\}}\,\mat{c}{z\\\bz}\,=\,e^{\f12 \veta\cdot\vtau}\,\mat{c}{z\\\bz}
\,,
\qquad
e^{\f12 \veta\cdot\vtau}=\mat{cc}{a & b \\ \bar{b} &\bar{a}}\in\SU(1,1)
\vspace*{2mm}\\
\left\{\veta\cdot\vcJ,\mat{c}{\bw\\ w}\right\}
\,=\,
\f i2\veta\cdot\vtau\,\mat{c}{\bw\\ w}
\,,\qquad
e^{\{\veta\cdot\vcJ\}}\,\mat{c}{\bw\\ w}\,=\,e^{\f12 \veta\cdot\vtau}\,\mat{c}{\bw\\ w}
\,.
\end{array}
\ee
where one should notice that the spinor components for $w$ are in reverse order compared to $z$. We use the notation $\veta\cdot\vcJ=\eta_{3}J_{3}-\eta_{1}K_{1}-\eta_{2}K_{2}$, with Lorentzian signature. We can also compute the exponentiated $\SU(1,1)$ action of the complex variables $t$ and $u$:
\be
e^{\{\veta\cdot\vcJ\}}\,\mat{c}{t\\\bt}\,=\,e^{\f12 \veta\cdot\vtau}\,\mat{c}{t\\\bt}
\,,\qquad
e^{\{\veta\cdot\vcJ\}}\,\mat{c}{u\\\bu}\,=\,e^{\f12 \veta\cdot\vtau}\,\mat{c}{u\\\bu}
\,,
\ee

Next, we would like to compute the $\SU(1,1)$ transformations induced by the 3-vector components on the 3-vector itself.
It is convenient to repackage the $\su(1,1)$ generators as a Hermitian 2$\times$2 matrix:
\be
\label{Mmatrix}
M
=
\mat{cc}{J_{3} & K_{-}\\ K_{+} & J_{3}}
\,.
\ee
We can reconstruct the Hermitian matrix $M$, either from the complex pair $(z,w)$ or the complex pair $(t,u)$:
\be
\label{Mdef}
M=
\mat{c}{z\\\bz}\,\mat{c}{z\\\bz}^{\dagger}
-
\mat{c}{\bw\\ w}\,\mat{c}{\bw\\ w}^{\dagger}
\,,
\ee
\be
\cM\equiv \mat{c}{u\\\bu}\,\mat{c}{t\\\bt}^{\dagger}
\,,\quad
\tr(\cM\tau_{3})=2i\cE
\,,\quad
M=\cM-\f12\tr(\cM\tau_{3})\,\tau_{3}
\,.
\ee
This allows to deduce the Poisson brackets of the $\su(1,1)$ generators with the matrix $M$ and write them as a matrix multiplication:
\be
\{\vcJ,M\}=\f i 2\,\big{(}
\vtau M -M \vtau^{\dagger}
\big{)}\,,
\ee
which is straightforward to exponentiate to the standard $\SO(2,1)$ action  on 3-vectors, i.e. the adjoint action of $\SU(1,1)$ on $\vcJ$:
\be
\label{conjugation}
e^{\{\veta\cdot\vcJ\}}M=GMG^{\dagger}
\quad\textrm{with}\,\,
G=e^{\f12 \veta\cdot\vtau}\in\SU(1,1)
\,,\quad
G\tau_{3}G^{\dagger}=\tau_{3}
\,.
\ee

\subsection{Generating shifts in the $\SU(1,1)$ Casimir: the whole $\so(3,2)$ Lie algebra}

Now that we have identified observables that generate boosts and rotations of the 3-vector $\vcJ$, we are interested in operators that would generate dilatations, or at least shifts, in the vector norm. The observable $\lambda$, introduced in \eqref{lambda} as canonically-conjugate to the norm $\cE$, could fit. However, it is not polynomial (nor analytic) in the complex variables $z$ and $w$, which would lead to quantization ambiguities. 

A solution to circumvent this problem is to introduce the whole $\so(3,2)$ Lie algebra generated by quadratic polynomials in the complex variables, as outlined in the appendix of \cite{Dupuis:2011wy}.
To start with, on top of the generator $J_{3}$, we introduce the other $\su(2)\sim\so(3)$ generators:
\be
J_{+}=\bz w
\,,\quad
J_{-}=z\bw=\overline{J}_{+}
\,,\qquad
\{J_{3},J_{\pm}\}=\mp i J_{\pm}
\,,\quad
\{J_{+},J_{-}\}=-2iJ_{3}
\ee
Then we can introduce another boost generator:
\be
K_{3}=-\f12(\bz\bw+zw)
\,,
\ee
so that the six generators, $J_{3},J_{\pm},K_{3},K_{\pm}$ form together a $\sl(2,\C)\sim\so(3,1)$ Lie algebra:
\be
\begin{array}{l}
\{K_{3},K_{\pm}\}=\pm i J_{\pm}
\,,\quad
\{K_{+},K_{-}\}=2i J_{3}
\,,\quad
\{J_{3},K_{\pm}\}=\mp i K_{\pm}
\,,\quad
\{K_{3},J_{\pm}\}=\mp i K_{\pm}
\,,
\vspace*{1mm}\\
\{J_{+},K_{-}\}=-2iK_{3}
\,,\quad
\{J_{-},K_{+}\}=2iK_{3}
\,,\quad
\{J_{3},K_{3}\}=\{J_{+},K_{+}\}=\{J_{-},K_{-}\}=0
\,.
\end{array}
\ee
Finally we further introduce another set of boost generators:
\be
L_{3}
=
\cE
=
\f i2\,\big{[}\bz\bw-zw\big{]}
\,,\quad
L_{+}
=
-\f {i}2\,\big{[}\bz^{2}+w^{2}\big{]}
\,,\quad
L_{-}
=\overline{L}_{+}
=
\f i2\,\big{[}z^{2}+\bw^{2}\big{]} \,.
\ee
Combining the $K$'s with $L$'s generates the special conformal transformations. In order to close the Lie algebra, we further have to introduce the dilatation generator,
\be
E=\f12\big{[}
z\bz+w\bw
\big{]}
\,.
\ee
The remaining Poisson brackets for the $\so(3,2)$ algebra are \footnotemark{}${}^{,}$\footnotemark{}:
\be
\begin{array}{l}
\{L_{3},L_{\pm}\}=\pm i J_{\pm}
\,,\quad
\{L_{+},L_{-}\}=2i J_{3}
\,,\quad
\{K_{3},L_{3}\}=-E
\,,\quad
\{K_{+},L_{-}\}=\{K_{-},L_{+}\}=-2E
\,,\vspace*{1mm}\\
\{J_{3},L_{\pm}\}=\mp i L_{\pm}
\,,\quad
\{L_{3},J_{\pm}\}=\mp i L_{\pm}
\,,\quad
\{J_{+},L_{-}\}=-2iL_{3}
\,,\quad
\{J_{-},L_{+}\}=2iL_{3}
\,,\vspace*{1mm}\\
\{K_{3},L_{\pm}\}=\{L_{3},K_{\pm}\}=\{K_{+},L_{+}\}=\{K_{-},L_{-}\}=0
\,,\quad
\{J_{+},L_{+}\}=\{J_{-},L_{-}\}=0
\,,\vspace*{1mm}\\
\{E,J_{a}\}=0
\,,\quad
\{E,K_{a}\}=L_{a}
\,,\quad
\{E,L_{a}\}=-K_{a}
\end{array}
\ee
%
\footnotetext{
The observable $E$ is actually a  Casimir of the $\su(2)$ Lie algebra. It allows to take the square-root of the quadratic $\su(2)$-Casimir:
\be
J_{3}^{2}+J_{1}^{2}+J_{2}^{2}=J_{3}^{2}+J_{+}J_{-}=E^{2}
\,.
\nn
\ee
It also gives the norm of the boost vectors $K$ and $L$, so that  the triplet of vectors $(\vJ,\vK,\vL)$ form an orthonormal basis of $\R^{3}$:
\be
\label{ortho}
\vJ^{2}=\vK^{2}=\vL^{2}=E^{2}
\,,
\quad
\vJ\cdot\vK=\vJ\cdot\vL=\vK\cdot\vL=0
\,,
\nn
\ee
where the scalar products are computed as $\vJ\cdot\vK=J_{3}K_{3}+\f12(J_{-}K_{+}+J_{+}K_{-})$ in the $z,\pm$ basis.
This means that the 3$\times$3 matrix $\f1E(\vJ,\vK,\vL)$ is orthogonal, which leads to equivalent orthonormality conditions between the lines of this matrix, in particular,
\be
\forall a\in\{1,2,3\}\,,\quad
J_{a}^2+K_{a}^2+L_{a}^2=E^2\,.
\nn
\ee
The six orthonormality conditions \eqref{ortho} on the ten $\so(3,2)$ generators reduces the data contained in the $\so(3,2)$ algebra down to the four real components  of the complex pair $(z,w)\in\C^{2}$. More precisely, an orthonormal $\R^{3}$-basis  $(\vJ,\vK,\vL)$ with given norm $E$ uniquely determines a pair of complex variables $(z,w)$.
}
\footnotetext{There are several $\su(1,1)$ Lie algebra within the $\so(3,2)$ algebra. Here, we have chosen to focus on the three generators $J_{3},K_{1},K_{2}$ to represent space-like 3-vectors. They form a $\su(1,1)$ Lie algebra with negative quadratic Casimir, $
J_{3}^2-K_{1}^2-K_{2}^2=-L_{3}^2\le 0$.
But other choices of generators, such as $E,K_{3},L_{3}$, lead to a $\su(1,1)$ Lie algebra with positive quadratic Casimir:
\be
\{E,K_{3}\}=L_{3}\,,\quad
\{E,L_{3}\}=-K_{3}\,,\quad
\{L_{3},K_{3}\}=E\,,
\qquad
E^2-K_{3}^2-L_{3}^2=+J_{3}^2\ge 0
\,.
\ee
Such a choice of presentation of $\su(1,1)$ allows to represent time-like 3-vectors and were used in \cite{Livine:2012mh} to obtain, after quantization, time-like unitary (lower weight) $\SU(1,1)$-representations in the context of quantum cosmology and define coherent cosmological states.}
All the generators $(J_{a},K_{a},L_{a},E)$ are real and will become Hermitian operators at the quantum level, thus leading to a unitary representation of the Lie group  $\SO(3,2)$.

Now, any $\so(3,2)$ operator which does not commute with $\cE=L_{3}$  generates shifts in the norm of the 3-vector. These operators, $J_{\pm},L_{\pm},K_{3},E$, have the following Poisson brackets wit $\cE$:
\be
\{\cE,L_{\pm}\}=\pm i J_{\pm}
\,,\quad
\{\cE,J_{\pm}\}=\mp i L_{\pm}
\,,\quad
\{\cE,K_{3}\}=E
\,,\quad
\{\cE,E\}=K_{3}
\ee
These can be organized into three separate $\sl_{2}$ algebra, each commuting with one direction in our original $\su(1,1)$ algebra generated by $J_{3},K_{1},K_{2}$:
\be
\begin{array}{l}
\left|\begin{array}{l}
\{\cE,J_{1}\}=L_{2}
\,,\,\,
\{\cE,L_{2}\}=J_{1}
\,,\,\,
\{J_{1},L_{2}\}=\cE\\
\{K_{1},J_{1}\}=\{K_{1},L_{2}\}=\{K_{1},\cE\}=0
\end{array}\right.
\vspace*{2mm}\\
\left|\begin{array}{l}
\{\cE,J_{2}\}=-L_{1}
\,,\,\,
\{\cE,L_{1}\}=-J_{2}
\,,\,\,
\{J_{2},L_{1}\}=-\cE\\
\{K_{2},J_{2}\}=\{K_{2},L_{1}\}=\{K_{2},\cE\}=0
\end{array}\right.
\vspace*{2mm}\\
\left|\begin{array}{l}
\{\cE,K_{3}\}=E
\,,\,\,
\{\cE,E\}=K_{3}
\,,\,\,
\{E,K_{3}\}=\cE\\
\{J_{3},K_{3}\}=\{J_{3},E\}=\{J_{3},\cE\}=0
\end{array}\right.
\end{array}
\ee
Let us focus, for example, on the $\sl_{2}$ algebra generated by the $\su(1,1)$ Casimir $\cE$, the boost generator $K_{3}$ and the dilatation generator $E$. We could use $E$ or $K_{3}$ to generate shifts in $\cE$:
\be
e^{\eta\,\{K_{3},\cdot\}}\cE=\cosh\eta\,\cE+\sinh\eta\,E
\,,\qquad
e^{\theta\,\{E,\cdot\}}\cE=\cos\theta\,\cE-\sin\theta\,K_{3}
\,.
\ee
These transformations mix $\cE$ with both $E$ and $K_{3}$ and do not straightforwardly shift $\cE$ to a initial fixed value to another a given value. At the quantum level, these will become unitary transformations and they will not map a state in a given irreducible $\SU(1,1)$-representation (with a certain value of the Casimir )to another irreducible $\SU(1,1)$-representation, but they will send a state with an initial given value of the Casimir onto a superposition state spread out on all possible values of the Casimir. Nevertheless, as a way to compensate this behavior, these transformations will not change the value of the $\su(1,1)$ rotation generator $J_{3}$.

On the other hand, we can consider the linear combination\footnotemark{} $(K_{3}+E)$, which has a simple Poisson bracket with $\cE$:
\be
\{\cE,K_{3}+E\}=K_{3}+E
\,.
\ee
\footnotetext{
We could also consider the combination $(E-K_{3})$. Actually the observable $\ln[(E-K_{3})/(E+K_{3})]$ has a canonical Poisson bracket with the $\su(1,1)$ Casimir $\cE$:
\be
\left\{
\cE,\,
\ln\f{E-K_{3}}{E+K_{3}}
\right\}=1
\,,
\nn
\ee
it would be perfect to generate clean shifts in $\cE$. However it is a non-analytic function of the complex variables and its quantization would require much more work in order to settle ordering and quantization ambiguity issues.
}
Upon quantization, this bracket becomes $[\hcE,\hK_{3}+\hE]=i(\hK_{3}+\hE)$, so we can use the operator $(\hK_{3}+\hE)$ to generate purely imaginary shifts in the $\su(1,1)$ Casimir. Starting with a state diagonalizing $\hcE$ with eigenvalue $\cE$ and, say,  $\hJ_{3}$ with eigenvalue $m$, we have:
\be
\hcE\,|\cE,m\ra=\cE\,|\cE,m\ra
\quad\Longrightarrow\quad
\hcE\,(\hK_{3}+\hE)\,|\cE,m\ra=(\cE+i)\,(\hK_{3}+\hE)\,|s,m\ra
\,.
\ee
This operator $(\hK_{3}+\hE)$, clearly well-defined at the quantum level, will not send unitary $\SU(1,1)$-representations onto unitary representations. One should not forget that it is actually  not a unitary operator but a Hermitian operator. This special feature  was explored in \cite{Girelli:2015ija}, where the authors investigated such a imaginary shift in the $\su(1,1)$ Casimir by  tensoring an irreducible unitary $\SU(1,1)$-representation with the fundamental 2-dimensional and non-unitary representation and then extracting from it irreducible unitary modules.


\section{Deformations of polyhedra in 2+1 dimensions}

We have described the phase space for a single space-like 3-vector. In this section, we consider $N$ independent space-like 3-vectors, provided with $N$ copies of the same Poisson bracket, bound together by a closure constraint:
\be
\cJ_{a}\equiv\sum_{i=1}^{N}\cJ^{i}_{a}=0
\,.
\ee
By the Minkowski theorem for convex polyhedra, this constraint ensures the existence of a unique convex polyhedron such that the 3-vectors $\vcJ^{i}\in\R^{1,2}$ are the normal vectors to its faces (see  e.g. \cite{Bianchi:2010gc} for an explicit reconstruction algorithm in the Euclidean case). Since the 3-vectors $\vcJ^{i}$  are assumed to be space-like, the polyhedron faces are time-like. And the norms $\cE^{i}=\pm|\vcJ^{i}|$ give the area of the faces.
We take this as the definition of the phase space of Lorentzian polyhedra with space-like normals.

The closure constraints are first class constraints, generating the invariance under global $\SU(1,1)$ transformations acting simultaneously on all $N$ normal vectors $\vcJ^{i}$. Taking the symplectic quotient by those constraints gives the phase space of Lorentzian polyhedra with space-like normals up to arbitrary $\SO(2,1)$ Lorentz transformations, i.e. Lorentzian polyhedron shapes, thus extending the Kapovich-Millson construction to the Lorentzian signature.

\smallskip

To be more precise, we introduce the Kapovich-Millson phase space for Lorentzian polyhedra with space-like nomals and fixed face areas as the symplectic quotient:
\be
\cP_{\{\cE_{i}\}_{i=1..N}}
\,=\,
\left\{
(z_{i},w_{i})\in\C^{2N}\right\}
//\R^{N}//\SU(1,1)
\,.
\ee
First, we impose the fixed area condition $\f i2(\bz_{k}\bw_{k}-z_{k}w_{k})=\cE_{k}$ for all $k$'s and consider the gauge orbits under the symmetry it generates, i.e. the $\lambda_{k}=\ln|u_{k}/t_{k}|=\ln[|z_{k}-\bw_{k}|/|z_{k}+\bw_{k}|]$ become pure gauge. The 3-vector components $\vJ^{k}_{a}$ indeed Poisson-commute with the $\cE_{k}$ and are good coordinates on this symplectic quotient. %
Then, we both impose the closure constraint, $\sum_{j}\vcJ^{j}=0$, and consider the gauge orbits under the $\SU(1,1)$ action that it generates. At the end of the day, the space $\cP_{\{\cE_{i}\}_{i=1..N}}$ has dimension $4N-2N-6=2(N-3)$.

\smallskip

Instead of fixing the face areas, we can remove the symplectic quotient by $\R^{N}$ and consider the full space, with dimension $4N-6$:
\be
\cP_{\{\cE_{i}\}_{i=1..N}}
\,=\,
\left\{
(z_{i},w_{i})\in\C^{2N}\right\}//\SU(1,1)
\,,
\ee
taking only the symplectic quotient by the closure constraint. Now, not only we have unfrozen the face areas $\cE_{k}$ but we also have their conjugate variables $\lambda_{k}$, defined earlier in \eqref{lambda} as a scale factor for the complex pair $(z_{k},w_{k})$. Since the 3-vectors $\vcJ^{k}$ determine a unique convex polyhedron with face areas $|\cE_{k}|=|\vcJ^{k}|$,  we now have one extra real variable  $\lambda_{k}$ per face (and signs according to proposition \ref{iso}). To follow the terminology introduced in the Euclidean case \cite{Freidel:2009ck,Freidel:2010aq,Freidel:2010bw,Livine:2013tsa}, we call these {\it Lorentzian framed polyhedra}, where we implicitly consider the $\lambda_{k}$'s as defining a scale or frame on each face.

\smallskip

Here we will identify the $\SU(1,1)$-invariant observables deforming  polyhedron shapes and show that they  generate $\SL_{N}(\R)$ transformations on polyhedra and that all polyhedra can be generated by such a transformation on an initial squashed polyhedron with only two non-trivial faces.

\subsection{$\SU(1,1)$-invariant observables and algebra of deformation}
\label{observables}

Following the work done in the Euclidean case \cite{Girelli:2005ii,Freidel:2009ck,Freidel:2010tt,Livine:2013tsa}, we introduce the following real quadratic combinations of the complex variables coupling pairs of normal vectors:
\be
\alpha_{ij}^{z}=i(\bz_{i}z_{j}-z_{i}\bz_{j})
\,,\quad
\alpha_{ij}^{w}=i(\bw_{i}w_{j}-w_{i}\bw_{j})
\,,\quad
\beta_{ij}=i(\bz_{i}\bw_{j}-z_{i}w_{j})\,,
\ee
which are very similar to the $\so(3,2)$ generators introduced in the previous section. It turns out that these provide all the quadratic  combinations Poisson-commuting with the closure constraint and thus invariant under $\SU(1,1)$:
\be
\{\cJ_{a},\alpha_{ij}^{z}\}=\{\cJ_{a},\alpha_{ij}^{w}\}=\{\cJ_{a},\beta_{ij}\}=0
\,.
\ee
First, we remark that the matrices $\alpha$ are anti-symmetric, $\alpha_{ij}^{z,w}=-\alpha_{ji}^{z,w}$, and that the diagonal components of the $\beta$ matrix give the faces' areas, $\beta_{ii}=2\cE^{i}$.

Second, we compute the Poisson brackets of those $\SU(1,1)$-invariant observables with the total boundary area $\cE\equiv\sum_{i}\cE^{i}$ and identify the generators of area-preserving deformations:
\be
\left|\begin{array}{lcl}
\alpha^{+}_{ij}\equiv\alpha_{ij}^{z}+\alpha_{ij}^{w}
\vspace*{1mm}\\
\beta^{S}_{ij}\equiv\beta_{ij}+\beta_{ji}
\end{array}
\right.
\,,\qquad
\left|\begin{array}{lcl}
\alpha^{+}_{ij}=-\alpha^{+}_{ji}
\vspace*{1mm}\\
\beta^{S}_{ij}=\beta^{S}_{ji}
\end{array}
\right.
\,,\qquad
\{\cE,\alpha^{+}_{ij}\}=\{\cE,\beta^{S}_{ij}\}=0\,.
\ee
The opposite combinations do not Poisson-commute with the total area and generate the area-changing deformations for polyhedron shapes:
\be
\left|\begin{array}{lcl}
\alpha^{-}_{ij}\equiv\alpha_{ij}^{z}-\alpha_{ij}^{w}
\vspace*{1mm}\\
\beta^{A}_{ij}\equiv\beta_{ij}-\beta_{ji}
\end{array}
\right.
\,,\qquad
\left|\begin{array}{lcl}
\alpha^{-}_{ij}=-\alpha^{-}_{ji}
\vspace*{1mm}\\
\beta^{A}_{ij}=-\beta^{A}_{ji}
\end{array}
\right.
\,,\qquad
\{\cE,\beta^{A}_{ij}\}=-\alpha^{-}_{ij}
\,,\quad
\{\cE,\alpha^{-}_{ij}\}=-\beta^{A}_{ij}
\,.
\ee
Finally, the two $\alpha^{z,w}$ matrices form two commuting  $\so(N)$ Lie algebra:
\be
\left|\begin{array}{l}
\{\alpha^{z}_{ij}, \alpha^{z}_{kl}\}
=
\delta_{jk}\alpha^{z}_{il}-\delta_{ik}\alpha^{z}_{jl}-\delta_{jl}\alpha^{z}_{ik}+\delta_{il}\alpha^{z}_{jk}
\vspace*{1mm}\\
\{\alpha^{w}_{ij}, \alpha^{w}_{kl}\}
=
\delta_{jk}\alpha^{w}_{il}-\delta_{ik}\alpha^{w}_{jl}-\delta_{jl}\alpha^{w}_{ik}+\delta_{il}\alpha^{w}_{jk}
\vspace*{1mm}\\
\{\alpha^{z}_{ij}, \alpha^{w}_{kl}\}=0
\end{array}\right.
\,,
\ee
and form, combined with the $\beta_{ij}$ observables a larger closed Lie algebra:
\be
\{\beta_{ij},\beta_{kl}\}=\delta_{ik}\alpha^{w}_{jl}+\delta_{jl}\alpha^{z}_{ik}
\,,\quad
\left|\begin{array}{l}
\{\alpha^{z}_{ij}, \beta_{kl}\}
=
\delta_{jk}\beta_{il}-\delta_{ik}\beta_{jl}
\vspace*{1mm}\\
\{\alpha^{w}_{ij}, \beta_{kl}\}
=
\delta_{jl}\beta_{ki}-\delta_{il}\beta_{kj}
\end{array}\right.
\,.
\ee
So, if we focus on area-preserving deformation generators, they form a closed $\gl_{N}(\R)$ Lie algebra:
\be
\label{gln}
\left|\begin{array}{lcl}
\{\alpha^{+}_{ij}, \alpha^{+}_{kl}\}
&=&
\delta_{jk}\alpha^{+}_{il}-\delta_{ik}\alpha^{+}_{jl}-\delta_{jl}\alpha^{+}_{ik}+\delta_{il}\alpha^{+}_{jk}
\vspace*{1mm}\\
\{\alpha^{+}_{ij}, \beta^{S}_{kl}\}
&=&
\delta_{jk}\beta^{S}_{il}-\delta_{ik}\beta^{S}_{jl}+\delta_{jl}\beta^{S}_{ik}-\delta_{il}\beta^{S}_{jk}
\vspace*{1mm}\\
\{\beta^{S}_{ij}, \beta^{S}_{kl}\}
&=&
\delta_{jk}\alpha^{+}_{il}+\delta_{ik}\alpha^{+}_{jl}+\delta_{jl}\alpha^{+}_{ik}+\delta_{il}\alpha^{+}_{jk}
\end{array}\right.
\,.
\ee
Below, we exponentiate this action and make explicit the area-preserving deformations on Lorentzian polyhedron shapes with space-like normal vectors.

\subsection{$\GL_{N}(\R)$ action on Lorentzian polyhedra}

It is much simpler to analyze the deformations using the complex variables $(u,t)$ instead of $(z,w)$ as introduced earlier in \eqref{utdef}, which allow to recover the 3-vector components, i.e. the $\su(1,1)$ generators, as simple products:
\be
u_{i}t_{i}=K_{-}^{i}
\,,\quad
\bu_{i}\bt_{i}=K_{+}^{i}
\,,\quad
u_{i}\bt_{i}=J_{3}^{i}+i\cE^{i}
\,.
\ee
Then we see that a linear recombination of the $u$'s can be compensated by a corresponding linear recombination of the $t$'s so as to keep invariant both  closure vector $\vcJ=\sum_{i}\vcJ^{i}$ and total area $\cE=\sum_{i}\cE^{i}$. Indeed, similarly to the Euclidean case \cite{Freidel:2009ck,Freidel:2010tt},  we consider linear transformations of the complex variables:
\be
\left|\begin{array}{lcl}
t_{i}&\mapsto & \sum_{j}M_{ij}t_{j} 
\\
u_{i}&\mapsto & \sum_{j}\tM_{ij}u_{j} 
\end{array}\right.
\quad\Longrightarrow\quad
\left|\begin{array}{rcccl}
K_{-}=\sum_{i}u_{i}t_{i}&\mapsto & \sum_{j,k}u_{j}\,({}^{t}\tM M)_{jk}\,t_{k} 
\\
J_{3}+i\cE=\sum_{i}u_{i}\bt_{i}&\mapsto & \sum_{j,k}u_{j}\,({}^{t}\tM\, \bM)_{jk}\,\bt_{k} 
\end{array}\right.
\,.
\ee
So the closure vector and total area are conserved if and only if ${}^{t}\tM M=\id={}^{t}\tM\, \bM$, which means that both $M$ and $\tM$ are real and invertible, $M\in\GL_{N}(\R)$ and $\tM={}^{t}M^{-1}$. In the following, we will consider the latter equality as the definition of $\tM$ in terms of $M$. The transformation laws for the original pair of canonical complex variables $(z,w)$ are easily deduced and mix $M$ and $\tM$:
\be
\begin{array}{rcl}
z=\f1{\sqrt2}(t+u)
&\longmapsto&
 \left(\f{M+\tM}2\right)z+\left(\f{M-\tM}2\right)\bw
\vspace*{1mm}\\
\bw=\f1{\sqrt2}(t-u)
&\longmapsto&
\left(\f{M+\tM}2\right)\bw+\left(\f{M-\tM}2\right)z
\end{array}
\,,
\ee
where we have kept implicit the notation of $z$ and $\bw$ as vectors in $\C^{N}$, and $M$ and $\tM$ as $N\times N$ matrices.
In general, $M$ and $\tM$ do not match, $M\ne\tM$, so that these transformations can be considered as Bogoliubov transformations mixing the two complex variables $z$ and $w$.

The transformation laws for the 3-vectors $\vcJ^{i}$ are more complicated and become non-linear. Moreover they involve the extra scaling parameters $\lambda_{i}$, introduced earlier in \eqref{lambda} in order to complete the information carried by a 3-vector $\vcJ\in\R^{1,2}$ into a full pair of complex variables $(z,w)\in\C^{2}$. Thus, the $\GL_{N}(\R)$ transformations truly act on framed Lorentzian polyhedron shapes. Indeed, both they will mix the 3-vectors $\vcJ^{i}$ and the attached scaling parameters $\lambda_{i}$.

\medskip

It is a straightforward exercise to check that the infinitesimal generators of those transformations are the $\alpha^{+}_{ij}$ and $\beta^{s}_{ij}$ introduced above and satisfying the Poisson algebra \eqref{gln}. Reversely, one can similarly check that the exponentiated action of those observables on the $u$'s and $t$'s reproduce this $\GL_{N}(\R)$ group action. In particular, the $\alpha^{+}$ generate the restricted case when $M\in \O(N)$ and thus $\tM=M$. In this special case, the transformations do not miw the two complex variables $z$ and $\bw$. 
More generally, an arbitrary invertible real matrix, $M\in\GL_{N}(\R)$, admits a unique polar decomposition, $M=QS$ with $Q\in\O(N)$ and $S$ symmetric positive definite. The $Q$ part is generated by the $\alpha^{+}_{ij}$ while the $S$ part is generated by the $\beta^{s}_{ij}$.

\begin{prop}
\label{orbits}
$\GL_{N}(\R)$-orbits:\\
We consider the $\GL_{N}(\R)$ action on Lorentzian framed (convex) polyhedra with $N$ time-like faces in $\R^{1,2}$.
Algebraically, these are defined as collections of pairs of complex variables $(t_{i},u_{i})$ satisfying the closure constraints, $\sum_{i}t_{i}u_{i}=0=\mathrm{Re}\,\sum_{i}\bt_{i}u_{i}$. Geometrically, these are are mapped onto sets of $N$ space-like 3-vectors $\vcJ^{i}$, satisfying the closure constraint $\sum_{i} \vcJ^{i}=0$, plus additional real parameters $\lambda_{i}$ and signs $\eps_{i},\varsigma_{i}=\pm$ for every faces  as proved in proposition \ref{iso}. These 3-vectors are the normal vectors to the polyhedron faces and the face areas are given by $\cE_{i}=\mathrm{Im}\,\bt_{i}u_{i}=\eps_{i}|\vcJ^{i}|$.

$\GL_{N}(\R)$ group elements act as invertible matrices on the complex pairs:
\be
\left|\begin{array}{lcl}
t_{i}&\mapsto & \sum_{j}M_{ij}t_{j} 
\\
u_{i}&\mapsto & \sum_{j}\tM_{ij}u_{j} 
\end{array}\right.
\,,
\quad \tM={}^{t}M^{-1}\,.
\ee
The $\GL_{N}(\R)$ action leaves invariant the total area $\cE=\sum_{i}\cE_{i}$. Assuming that $\cE\ne 0$, the $\GL_{N}(\R)$ is cyclic
 and can always map arbitrary closed configurations to a totally squashed configuration:
\be
(u^{0}_{1},t^{0}_{1})
=
\sqrt{\f\cE2}\,(1,-i)
\,,\qquad
(u^{0}_{2},t^{0}_{2})
=
i\,(u^{0}_{1},t^{0}_{1})
=
\sqrt{\f\cE2}\,(i,1)
\,,\qquad
(u^{0}_{j},t^{0}_{j})=0\,,\,\,\forall j\ge 3\,,
\ee
which corresponds to a flattened framed  polyhedron with $\vcJ^{1}=\f{\cE}2\,\hat{e}_{2}$ pointing in the $y$-direction, $\vcJ^{2}=-\vcJ^{1}$ and $\vcJ^{j}=0$ for all $j\ge 3$. This means that an arbitrary closed configuration can be decomposed as:
\be
\left|\begin{array}{lcl}
t_{i}=\sqrt{\f\cE2}\,(M_{i2}-iM_{i1})
\\
u_{i}=\sqrt{\f\cE2}\,(\tM_{i1}+i\tM_{i2})
\end{array}\right.
\,,\quad
M\in\GL_{N}(\R)
\,,\quad
\tM={}^{t}M^{-1}\,.
\ee
Finally, the stabilizer subgroup of the squashed configuration $\{(u^{0}_{i},t^{0}_{i})\}_{i}$ is $\GL_{N-2}(\R)$, so the space of closed configurations, i.e. the space of framed polyhedra, with fixed non-vanishing total area $\eps$ is isomorphic to the coset $\GL_{N}(\R)/\GL_{N-2}(\R)$. Its dimension is $N^{2}-(N-2)^{2}=4N-3-1$ as expected, with the $-3$ coming from the closure constraint and the $-1$ from the fixed total area condition.

\end{prop}

Starting with an arbitrary collections to complex pairs $\{(t_{i},u_{i})\}_{i=1..N}$ satisfying the closure constraints, the goal is to show that it can always be in a canonical form as the action of a $\GL_{N}(\R)$ group element $M$ on a reference closed configuration  $\{(t_{i}^0,u_{i}^0)\}_{i=1..N}$ such that $t_{j}^0=u_{j}^0=0$ for all $j\ge 3$. One way to see this as a diagonalization process. 
Let us combine the antisymmetric matrix $\alpha^{+}_{ij}$ with the symmetric matrix $\beta^{s}_{ij}$ into a single real $N\times N$ matrix $\Delta$:
\be
\Delta_{ij}=\alpha^{+}_{ij}+\beta^{s}_{ij}
=
2i(\bu_{i}t_{j}-u_{i}\bt_{j})
\,\in\R
\,.
\ee
This whole matrix is invariant under global $\SU(1,1)$ transformations acting simultaneously on all pairs $(t_{i},u_{i})$.
Using the closure constraints, $\sum_{j}t_{j}u_{j}=0$ and $\sum_{j}\bt_{j}u_{j}=i\cE$, we can compute the trace of this matrix in terms of the total area $\cE$ and show that it satisfy a very simple quadratic polynomial equation:
\be
\tr\,\Delta=4\cE
\,,\quad
\Delta^{2}=2\cE\,\Delta\,.
\ee
Assuming that $\cE\ne 0$, this shows that, first, $\Delta$ is diagonalizable\footnotemark{} with eigenvalues 0 and $2\cE$, second, that $\Delta$ is a rank-2 matrix, with the eigenvalue 0 having degeneracy $(N-2)$ and the eigenvalue $2\cE$ having degeneracy 2:
\be
\exists \tM\in\GL_{N}(\R)\,,\,\,\textrm{such that}\quad
\Delta
=
2\cE\,\tM\,
\mat{c|c}{
\id_{2}& \\ \hline
 & 0_{N-2}
}
\,\tM^{-1}
\,.
\ee
We have two sectors: the non-trivial sector $j=1,2$  for $\Delta_{jk}$, which corresponds to non-trivial pairs $\{(t_{i}^0,u_{i}^0)\}_{i=1,2}$  and the $j\ge 3$ sector with vanishing matrix elements and corresponding to vanishing polyhedron faces.
\footnotetext{
We further introduce the real matrix $\Theta_{ij}\equiv2(\bu_{i}t_{j}+u_{i}\bt_{j})$, whose matrix elements are not $\SU(1,1)$-invariant but which satisfies the matrix identities:
\be
\tr\,\Theta=0
\,,\quad
\Delta \Theta=\Theta\Delta=2\cE\,\Theta
\,,\quad
\Theta^{2}=-2\cE\Delta
\,.
\nn
\ee
Since $\Delta$ and $\Theta$ commute with each other, this implies that the two real matrices can be simultaneously put in a canonical form:
\be
\exists \tM\in\GL_{N}(\R)\,,\,\,\textrm{such that }
\Delta
=
2\cE\,\tM\,
\mat{c|c}{
\id_{2}& \\ \hline
 & 0_{N-2}
}
\,\tM^{-1}
\,\textrm{ and }\,
\Theta=
2\cE\,\tM\,
\mat{cc|c}{
0& 1 & \\
-1 & 0 \\ \hline
& & 0_{N-2}
}
\,\tM^{-1}
\,.
\nn
\ee
The change of basis matrix $\tM$ is actually defined up to a  $\SO(2)\times\GL_{N-2}(\R)$ freedom.
One can study which complex pairs $\{(t_{i}^0,u_{i}^0)\}_{i}$ lead to the canonical forms of $\Delta$ and $\Theta$. Focusing on the non-trivial sector, the first two complex pairs satisfy the simple quadratic equations:
\be
u_{1}^0\bt_{1}^0=i\f\cE2
\,,\quad
u_{1}^0\bt_{2}^0=\f\cE2
\,,\quad
u_{2}^0\bt_{1}^0=-\f\cE2
\,,\quad
u_{2}^0\bt_{2}^0=i\f\cE2
\,,
\nn
\ee
implying that we can compute $t_{1},u_{2},t_{2}$ in terms of $u_{1}$:
\be
(u^{0}_{1},t^{0}_{1})=\sqrt{\f\cE2}\,(u,-i\bu^{-1})
\,,\qquad
(u^{0}_{2},t^{0}_{2})=i\,(u^{0}_{1},t^{0}_{1})
=\sqrt{\f\cE2}\,(iu,\bu^{-1})
\,,\qquad
u\in\C
\,.
\nn
\ee
Assuming that all the other pairs vanish, $(u^{0}_{j},t^{0}_{j})=0$ for $j\ge 3$, we can compute the action of a $M\in\GL_{N}(\R)$ group element on this canonical form:
\be
\{(u_{i},t_{i})\}_{i}=M\triangleright\{(u_{i}^0,t_{i}^0)\}_{i}
=\{((\tM u^0)_{i},(Mt^0)_{i})\}_{i}
\quad\Rightarrow\quad
\left|
\begin{array}{l}
u_{i}=(\tM_{i1}+i\tM_{i2})\,u
\\
t_{i}=(M_{i2}-iM_{i1})\,\bu^{-1}
\end{array}
\right.
\nn
\ee
We can compensate the phase of $u$ by a $\SO(2)$ transformation, and then absorb its modulus in a re-definition of the matrix $M$. This means that we can set $u=1$, which leads to the ansatz given in the proposition \ref{orbits} above.
}

Let us now prove the proposition.

\begin{proof}

Assuming that $\cE\ne 0$, we start by distinguishing the real and imaginary parts of the $u$'s and $t$'s:
\be
u_{i}=\sqrt{\f\cE2}\,(A_{i}+iB_{i})\,,\quad t_{i}=\sqrt{\f\cE2}\,(C_{i}-iD_{i})\,,\qquad
A,B,C,D\in\R^N\,,
\ee
where $A$, $B$, $C$ and $D$ are real vectors of dimension N. The closure constraints, combine with the total fixed area condition, give orthonormality relations between those vectors:
\be
\label{closureABCD}
\left|
\begin{array}{l}
\sum_{i}u_{i}t_{i}=0\\
\sum_{i}u_{i}\bt_{i}=i\cE
\end{array}
\right.
\qquad\Longrightarrow\quad
\left|
\begin{array}{l}
A\cdot D = B \cdot C =0
\\
A\cdot C = B \cdot D =1
\end{array}
\right.
\,.
\ee
We would like to identify $A$ and $B$ as the first two columns of a matrix $\tM$ and $D$ and $C$ as the first two columns of a matrix $M$, such that $M$ and $\tM$ are both invertible and related by $M={}^t\tM^{-1}$, i.e.:
\be
\label{conditions}
\textrm{Does there exist } M,\tM\in\GL_{N}(\R)
\quad\textrm{such that}\quad
\left|
\begin{array}{l}
(A_{i},B_{i})=(\tM_{i1},\tM_{i2})
\\
(D_{i},C_{i})=(M_{i1},M_{i2})
\end{array}
\right.
\quad\textrm{and} \quad M={}^t\tM^{-1}\,\,\,?
\ee
Let us first point out that if $M={}^t\tM^{-1}$ then the first columns of the matrices automatically satisfy the required closure constraints.
\be
 M={}^t\tM^{-1}
 \quad\Longleftrightarrow\quad
 \forall j,k\,,\,\,
 \sum_{i}M_{ij}\tM_{ik}=\delta_{jk}
 \quad\Longrightarrow\quad
 \left|
\begin{array}{l}
 \sum_{i}M_{i1}\tM_{i2}= \sum_{i}M_{i2}\tM_{i1}=0
\\
 \sum_{i}M_{i1}\tM_{i1}= \sum_{i}M_{i2}\tM_{i2}=1
\end{array}
\right.
 \,.
\ee
We need to show the reverse. Having the $u_{i}$'s fixes the first two columns of $\tM$ and we need to determine the remaining columns, i.e. we are looking for vectors $\tM_{i3}$, .., $\tM_{iN}$ such that: 1. together with $\tM_{i1}=A_{i}$ and $\tM_{i2}=B_{i}$ they form a basis of $\R^N$ and 2. they are orthogonal to both $C$ and $D$ in order to ensure that $M={}^t\tM^{-1}$.

Let us call $V^1$ the unit vector in the direction $A$ in the vector space $\R^N$, i.e. we write $A= a V^1$ wih $V^1\cdot V^1=1$. Then $B$ can not be collinear to $A$, else it would violate the closure constraints \eqref{closureABCD}. So we can write $B=b_{1}V^1+b_{2}V^2$, where $V^2$ is normed, $V^2\cdot V^2=1$, and orthogonal to $V^1$, i.e. $V^1\cdot V^2=0$. Next, $D$ is orthogonal to $A=V^1$ but has a non-vanishing scalar product with $B$, we can thus introduce a unit vector $V^3$ orthogonal to both $V^1$ and $V^2$ to write it as $D=V^2+dV^3$. Similarly for $C$, we introduce a unit vector $V^4$, orthogonal to $V^1$, $V^2$ and $V^3$, to write:
\be
C=V^1-\f{b_{1}}{b_{2}}V^2+c_{3}V^3+c_{4}V^4
\,.
\nn
\ee
Let us now complete the $V$ into an orthonormal basis $V^1,.., V^N$ in $\R^N$. It is then straightforward to identify a possible solution for the remaining columns of $\tM$:
\be
\tM_{\bullet 3}=V^3-dV^2
\,,\quad
\tM_{\bullet 4}=V^4-\f{c^4}{c_{3}+d\f{b_{1}}{b_{2}}}(V^3-dV^2)
\,,\quad
\tM_{\bullet k}=V^k \,\,\forall k\ge 5\,.
\ee
Since these columns are  linearly-independent by definition, this matrix $\tM$ and its transpose inverse  satisfy all the required conditions \eqref{conditions}.

\end{proof}

This proposition concludes the analysis of the orbits under $\GL_{N}(\R)$ of the sets of $N$ complex pairs $(u_{i},t_{i})$  satisfying the closure constraints: we have a single orbit for each value for the total area $\cE$ when $\cE$ is non-vanishing\footnotemark.
\footnotetext{The case when $\cE=0$ is probably much trickier to classify. This happens  when the signs $\eps_{i}$ are not all positive and allow $\sum_{i}\cE_{i}=\sum_{i}\eps|\vcJ^i|=0$. Intuitively, one would like a similar result as when $\cE\ne0$ acting with $\GL_{N}(\R)$ on squashed configurations, as in proposition \ref{orbits}:
\be
(u^{0}_{1},t^{0}_{1})
=
\rho\,(1,-i)
\,,\qquad
(u^{0}_{2},t^{0}_{2})
=
\rho\,(1,i)
\,,\qquad
(u^{0}_{j},t^{0}_{j})=0\,,\,\,\forall j\ge 3\,,
\nn
\ee
which define the same squashed polyhedron (i.e. same normal vectors $\vcJ^i$) but with $\eps_{2}=-$. A priori, $\GL_{N}(\R)$-orbits would then be labeled by $\rho$. This case nevertheless remains to be studied rigorously.
}
This means that the present case of Lorentzian polyhedra with space-like normals is very similar to the Euclidean case where one could explore the space of framed polyhedra at fixed total area by acting with $\U(N)$ transformations on a squashed polyhedron with only two non-trivial faces \cite{Freidel:2009ck,Livine:2013tsa}. In the Euclidean case, the space of framed polyhedra at fixed total area was then isomorphic to the cose $\U(N)/\U(N-2)$ while it is here, in the Lorentzian case, isomorphic to $\GL_{N}(\R)/\GL_{N-2}(\R)$. It is a very strong result that one can generate all possible polyhedron shapes by acting with well-defined $\GL_{N}(\R)$ transformations on a squashed polyhedron.

This provides a natural measure on the space of Lorentzian polyhedra with space-like normals and opens the door to the analysis of concentration of measure phenomena on $\GL_{N}(\R)$ and typicality as in the Euclidean case \cite{Livine:2013tsa,Anza:2016fix,Anza:2017dkd}.

\subsection{Closing arbitrary configurations to generate polyhedra}

Up to now in this section, we have assumed that our initial complex variables satisfy the closure constraints and we have described deformations compatible with those closure constraints. We distinguished area-preserving and area-changing deformations and we showed that the area preserving deformations form a $\GL_{N}(\R)$ group acting cyclically on Lorentzian polyhedra at fixed area.
Here we will tackle deformations that allow to go in and out of the closure constraints, with the goal of finding a systematic way to deform an arbitrary non-closed configuration into a a closed configuration admitting a geometrical interpretation as a Lorentzian polyhedron. 

Starting with $N$ complex pairs $(z_{i},w_{i})$, we explained in the first section that the 3-vector components are a sum of independent contributions from the $z$'s and from the $w$'s. Starting from \eqref{Mdef}, we have:
\be
M
=
\mat{cc}{J_{3} & K_{-}\\ K_{+} & J_{3}}
=
\sum_{i} M_{i}
\,,\qquad
M_{i}
=
\mat{cc}{J_{3}^{i} & K_{-}^{i}\\ K_{+}^{i} & J_{3}^{i}}
=
\mat{c}{z_{i}\\\bz_{i}}\,\mat{c}{z_{i}\\\bz_{i}}^{\dagger}
-
\mat{c}{\bw_{i}\\ w_{i}}\,\mat{c}{\bw_{i}\\ w_{i}}^{\dagger}
\,.
\ee
A $\SU(1,1)$ transformation acts simultaneously on all the complex pairs $(z_{i},w_{i})$ as:
\be
G=\mat{cc}{a & b \\ \bar{b}& \bar{a}}
\in\SU(1,1)
\,,\qquad
\mat{c}{z_{i}\\\bz_{i}}
\,\longmapsto\,
G\,\mat{c}{z_{i}\\\bz_{i}}
\,,\quad
\mat{c}{\bw_{i}\\ w_{i}}
\,\longmapsto\,
G\,\mat{c}{\bw_{i}\\ w_{i}}
\,,\quad
M_{i}
\,\longmapsto\,
GM_{i}G^{\dagger}
\,,
\ee
leading to an action by conjugation on the global $\su(1,1)$ generators, $M\,\mapsto GMG^{\dagger}$. This clearly leaves invariant the norm of the 3-vector, $|\vcJ|^{2}=J_{3}^{2}-K_{-}K_{+}=\det\,M$. We need to extend these transformations to allow to go from a non-closed configuration to a closed configuration.

It seems natural to distinguish the contribution coming from the $z$'s and the one coming from the $w$'s:
\be
M_{i}
=
\mat{c}{z_{i}\\\bz_{i}}\,\mat{c}{z_{i}\\\bz_{i}}^{\dagger}
-
\mat{c}{\bw_{i}\\ w_{i}}\,\mat{c}{\bw_{i}\\ w_{i}}^{\dagger}
=M_{i}^{z}-M_{i}^{w}
\,,\qquad
M=\sum_{i}^{N}M_{i}^{z}-\sum_{i}^{N}M_{i}^{w}=M^{z}-M^{w}
\,.
\ee
This splits the 3-vector $\vcJ$ into two parts, $M^z$ and $M^w$.
Then we extend $\SU(1,1)$ transformations to pairs $(G, \tG)\in\SU(1,1)^{\times 2}$ a priori acting differently on the $z$'s and $w$'s:
\be
\label{doubleSU11}
(G, \tG)\triangleright \left(\mat{c}{z_{i}\\\bz_{i}}\,,\, \mat{c}{\bw_{i}\\ w_{i}}\right)
=
\left(G\,\mat{c}{z_{i}\\\bz_{i}}\,,\,\tG\, \mat{c}{\bw_{i}\\ w_{i}}\right)
\,,\qquad
M^{z}\mapsto GM^{z}G^{\dagger}
\,,\quad
M^{w}\mapsto \tG M^{w}\tG^{\dagger}
\,.
\ee
This allows to act with $\SO(2,1)$ Lorentz transformations (rotations and boosts) independently on the two parts of the 3-vector $\vcJ$. Since these transformations do not change the norms, $\det M^z$ and $\det M^w$, this is not enough to have them cancel each other. And we further introduce a rescaling transformation by $\tau\in\R$ that  changes the relative norms of the two parts $M^z$ and $M^w$:
\be
\label{rescaling}
\tau\triangleright
 \left(\mat{c}{z_{i}\\\bz_{i}}\,,\, \mat{c}{\bw_{i}\\ w_{i}}\right)
=
 \left(e^{\f\tau2}\mat{c}{z_{i}\\\bz_{i}}\,,\, e^{-\f\tau2}\mat{c}{\bw_{i}\\ w_{i}}\right)\,.
\ee
It is easier to follow the procedure geometrically. The $z_{i}^s$ define $N$ future-oriented null-like vectors. The sum of those vectors, encoded in the Hermitian matrix $M^z$, is a future-oriented time-like vector. This vector can not vanish except if all the $z_{i}$'s vanish. Similarly, the $w_{i}$'s define $N$ past-oriented null-like vectors and their sum, encoded in the Hermitian matrix $M^w$, is a past-oriented time-like vector. First, there exists a unique rescaling \eqref{rescaling} such that the two time-like vectors end up with the same norm. Second there exists a unique $\SU(1,1)$ transformation $(\id,\Lambda)\in\SU(1,1)^{\times 2}$, acting as defined above in \eqref{doubleSU11}, up to an arbitrary boost in the direction $M^w$, such that it boosts the past-oriented time-like vector to exactly cancel the future-oriented time-like vector.

At the end of the day, this establishes that the space of collections of complex pairs $\{(z_{i},w_{i})\}_{i=1..N}\in\C^{2N}$, excluding the degenerate cases where all the $z$'s or all the $w$'s are equal to 0, quotiented by the $\SU(1,1)^{\times 2}$ action and the rescaling group $\R$ is isomorphic to the space of framed Lorentzian polyhedron shapes (with space-like normals) up to $\SU(1,1)$ transformations and up to rescaling to the total area:
\be
(\C^{2N})^*/(\SU(1,1)^{\times 2}\times \R)
\sim
\Big{(}(\C^{2N})^*//\SU(1,1)\Big{)}/\R
\sim
\Big{(}\GL_{N}(\R)/\GL_{N-2}(\R)\Big{)}/\SU(1,1)
\,,
\ee
where $(\C^{2N})^*$ stands for $\C^2N$ excluding all degenerate pairs $\{(z_{i},0)\}_{i=1..N}$ and $\{(0,w_{i})\}_{i=1..N}$.

\section{The fine structure of the space of $\SU(1,1)$ intertwiners}

In this section, we  tackle the quantization of the Lorentzian framed polyhedron phase space. We start with a single space-like vector and quantize the $\su(1,1)$ algebra as a pair of harmonic oscillators. This leads to the principal continuous series of $\SU(1,1)$-representations, dubbed space-like representations. Then $N$-face polyhedra (with space-like normals) are quantized as $\SU(1,1)$-invariant states in the tensor products of $N$ such representations. We show that this space of $\SU(1,1)$-intertwiners carries a representation of the $\GL(N,\R)$ group. The $\gl_{N}(\R)$ generators and $\GL(N,\R)$ unitary transformations provide a whole toolbox of operators, beyond the mere total area observable, to distinguish and finely deform quantum Lorentzian polyhedra.

\subsection{Space-like quantum vectors and $\SU(1,1)$-representations}

Let us quantize the phase space for a single Lorentzian space-like vector. We quantize the pair of complex variables $(z,w)$ as a pair of harmonic oscillators:
\be
[a,\ad]=[b,\bd]=1\,.
\ee
Following the classical definitions \eqref{su11}, we quantize the components of the Lorentzian 3-vector:
\be
\hJ_{3}=\f12(\ad a -\bd b)\,,
\quad
\hK_{+}=\f12\big{(}
(\ad)^{2}-b^{2}
\big{)}\,,
\quad
\hK_{-}=\hK_{+}^{\dagger}=
\f12\big{(}
a^{2}-(\bd)^{2}
\big{)}\,,
\ee
which satisfy as expected the $\su(1,1)$ Lie algebra commutation relations:
\be
[\hJ_{3},\hK_{\pm}]=\pm \hK_{pm}
\,,
\quad
[\hK_+,\hK_-]=2\hJ_{3}
\,.
\ee
Further, we compute the quadratic Casimir and express it as the square of an energy operator:
\be
\hC=J_{3}^{2}-\f12(K_{+}K_{-}+K_{-}K_{+})
=-\big{(}\hcE^{2}+\f14\big{)}
\,,
\quad
[\hC,\hJ_{3}]=[\hC,\hK_{\pm}]=0
\qquad\textrm{with}\quad
\hcE=\f i2(\ad \bd-ab)\,.
\ee
$\hcE$ is a standard squeezing operator acting on pairs of harmonic oscillators as creation and annihilation of entangled quanta of energy. It is an Hermitian operator with real spectrum.

The Hilbert space consists in the direct sum of two copies of the harmonic oscillator Hilbert space, with the usual basis $|n_{1},n_{2}\ra$ labeled by the number of quanta of the two oscillators, $n_{1,2}\in\N$. In this basis, the $\su(1,1)$ generators have a straightforward action:
\be
\cH_{HO}^{\otimes 2}=\bigoplus_{n_{1},n_{2}\in\N}\C\,|n_{1},n_{2}\ra\,,\quad
\left|\begin{array}{lcl}
\hJ_{3} \,|n_{1},n_{2}\ra
&=&
\f12(n_{1}-n_{2})\,|n_{1},n_{2}\ra
\vspace*{2mm}\\
\hK_{+} \,|n_{1},n_{2}\ra
&=&
\f12\Big{[}
\sqrt{(n_{1}+1)(n_{1}+2)} \,|n_{1}+2,n_{2}\ra
-
\sqrt{n_{2}(n_{2}-1)} \,|n_{1},n_{2}-2\ra
\Big{]}
\vspace*{1mm}\\
\hK_{-} \,|n_{1},n_{2}\ra
&=&
-\f12\Big{[}
\sqrt{(n_{2}+1)(n_{2}+2)} \,|n_{1},n_{2}+2\ra
-
\sqrt{n_{1}(n_{1}-1)} \,|n_{1}-2,n_{2}\ra
\Big{]}
\end{array}\right.
\,.\nn
\ee
One would like to go from this basis to the usual basis diagonalizing the $\su(1,1)$ Casimir and the rotation generator $\hJ_{3}$. One expects to recover the unitary representation of the $\SU(1,1)$ Lie group from the principal continuous series. These are labeled by the eigenvalue $s\in\R$ of the Casimir $\hcE$ and the eigenvalue $m\in\f\Z2$ of  the rotation generator $\hJ_{3}$, plus a parity $\eps=\pm$ which we usually keep implicit, with the standard action of the $\su(1,1)$ generators:
\be
\left|\begin{array}{lcl}
\hcE \,|s,m\ra
&=&
s \,|s,m\ra
\vspace*{1mm}\\
\hJ_{3} \,|s,m\ra
&=&
m \,|s,m\ra
\vspace*{1mm}\\
\hK_{+}  \,|s,m\ra
&=&
\sqrt{m(m+1)-\mathfrak{C}} \,|s,m+1\ra
\vspace*{1mm}\\
\hK_{-}  \,|s,m\ra
&=&
\sqrt{m(m-1)-\mathfrak{C}} \,|s,m-1\ra
\end{array}\right.
\,,
\ee
with the value of the quadratic Casimir given in terms of $s$ as:
\be
\widehat{\mathfrak{C}}=\hJ_{3}^{2}-\f12(\hK_{+}\hK_{-}+\hK_{-}\hK_{+})=-\Big{[}\hcE^{2}+\f14\Big{]}
\,,\qquad
\widehat{\mathfrak{C}}\,|s,m\ra=-\Big{[}s^{2}+\f14\Big{]}\,|s,m\ra
\,.
\ee
Even representations, with $\eps=+$, decompose over basis states with $m\in\Z$ while odd representations, with $\eps=-$, decompose over $m\in\Z+\f12$.

For the correspondence between the two basis, we first have $n_{1}-n_{2}=2m$. So $m$ translates to an energy shift between the two oscillators. Diagonalizing the squeezing operator $\hcE$ is more involved. Assuming without loss of generality that $m\ge 0$, we write the eigenvalue equation:
\beq
\hcE\,|s,m\ra
&=&
\hcE\,\sum_{n\in\N}\alpha_{n}\,|n+2m,n\ra=\sum_{n\in\N}s\alpha_{n}\,|n+2m,n\ra
\\
&=&
\f i2
\sum_{n\in\N}
\alpha_{n}\sqrt{(n+1)(n+2m+1)}\,|n+1+2m,n+1\ra
-
\alpha_{n}\sqrt{n(n+2m)}\,|n-1+2m,n-1\ra
\nn
\eeq
which translates to a 2nd degree recursion relation:
\be
\left|
\begin{array}{lcl}
2is\,\alpha_{0}
&=&
\alpha_{1}\sqrt{2m+1}
\\
2is\,\alpha_{n}
&=&
\alpha_{n+1}\sqrt{(n+1)(n+2m+1)}-\alpha_{n-1}\sqrt{n(n+2m)}
\end{array}\right.
\ee
This leads to a single vector (up to normalizing the initial condition $\alpha_{0}$), which may be exactly expressed in terms of hypergeometric functions. The asymptotics of the solutions at large $n$ is simpler to obtain:
\be
\alpha_{n}\underset{n\rightarrow +\infty}{\sim}
n^{is-\f12}
\ee
A question that we postpone to later investigation is to write coherent states \`a la Perelomov \cite{Perelomov:1986tf}, with minimal uncertainty relations and coherent under the action of the $\SU(1,1)$ group, which would define semi-classical space-like 3-vectors.

\subsection{$\SU(1,1)$ Intertwiners and Quantum Lorentzian Polyhedra}

Now we take $N$ copies of the Hilbert space of pairs of harmonic oscillators, $(\cH_{HO}\otimes \cH_{HO})^{\otimes N}$, thus working with tensor products of $N$ $\SU(1,1)$-representations. The classical closure constraint, $\sum_{i=1}^{N}\vcJ^{i}=0$ translates at the quantum level into the requirement of invariance under the global  $\SU(1,1)$ action acting simultaneously on the $N$ representations. We take this as our definition of the space of quantum (framed) Lorentzian polyhedra in terms of $\SU(1,1)$ intertwiners:
\be
\cH
\equiv\textrm{Inv}_{\SU(1,1)}\,\Big{[}
(\cH_{HO}\otimes \cH_{HO})^{\otimes N}
\Big{]}
\,.
\ee
We can raise the classical $\SU(1,1)$-invariant observables, introduced earlier in section \ref{observables}, to quantum deformation operators acting on our space of quantum polyhedra $\cH$. Since they are quadratic polynomials in the complex variables, it is direct to quantize them as quadratic operators in the harmonic oscillators' creation and annihilation operators:
\be
\halpha^{z}_{ij}=i\Big{[}\ad_{i}a_{j}-a_{i}\ad_{j}+\delta_{ij}\id\Big{]}\,,\quad
\halpha^{w}_{ij}=i\Big{[}\bd_{i}b_{j}-b_{i}\bd_{j}+\delta_{ij}\id\Big{]}\,,\quad
\hbeta_{ij}=i\Big{[}\ad_{i}\bd_{j}-a_{i}b_{j}\Big{]}\,.
\ee
The extra term $+\delta_{ij}\id$ is due to quantum ordering and ensures that the operators $\halpha_{ij}^{a,b}$ are Hermitian and  antisymmetric under the exchange $i\leftrightarrow j$. In particular, the operators on the diagonal  vanish, $\halpha^{a,b}_{ii}=0$.
The operators $\hbeta_{ij}$ are also Hermitian and the diagonal components give the area of the quantum polyhedron face, i.e. the $\su(1,1)$ Casimir of the $N$ representations, $\hbeta_{ii}=2\hcE^{i}$.

It is straightforward to check that these operators commute with the $\su(1,1)$ generators, $J_{3}=\sum_{i}J_{3}^{i}$ and $K_{\pm}=\sum_{i}K_{\pm}^{i}$:
\be
\Big{[}
\cJ_{a}\,,\,
\halpha^{z}_{ij}
\Big{]}
=
\Big{[}
\cJ_{a}\,,\,
\halpha^{w}_{ij}
\Big{]}
=
\Big{[}
\cJ_{a}\,,\,
\hbeta_{ij}
\Big{]}
=
0
\,,
\ee
so that they are legitimate operators on the Hilbert space  $\cH$ of $\SU(1,1)$-invariant states.
Moreover they still form a closed Lie algebra at the quantum level:
\be
\left\{\begin{array}{lcl}
[\halpha^{z}_{ij}\,,\,\halpha^{z}_{kl}]
&=&
i\big{(}
\delta_{jk}\halpha^{z}_{il}-\delta_{ik}\halpha^{z}_{jl}-\delta_{jl}\halpha^{z}_{ik}+\delta_{il}\halpha^{z}_{jk}
\big{)}
\vspace*{1mm}\\
{[}\halpha^{w}_{ij}\,,\,\halpha^{w}_{kl}{]}
&=&
i\big{(}
\delta_{jk}\halpha^{w}_{il}-\delta_{ik}\halpha^{w}_{jl}-\delta_{jl}\halpha^{w}_{ik}+\delta_{il}\halpha^{w}_{jk}
\big{)}
\vspace*{1mm}\\
{[}\halpha^{z}_{ij}\,,\,\halpha^{w}_{kl}{]}
&=&
0
\vspace*{1mm}\\
{[}\hbeta_{ij}\,,\,\hbeta_{kl}{]}
&=&
i\big{(}
\delta_{jl}\halpha^{z}_{ik}+\delta_{ik}\halpha^{w}_{jl}
\big{)}
\vspace*{1mm}\\
{[}\halpha^{z}_{ij}\,,\,\hbeta_{kl}{]}
&=&
i(\delta_{jk}\hbeta_{il}-\delta_{ik}\hbeta_{jl})
\vspace*{1mm}\\
{[}\halpha^{w}_{ij}\,,\,\hbeta_{kl}{]}
&=&
i(\delta_{jl}\hbeta_{ki}-\delta_{ik}\hbeta_{kj})
\end{array}\right.
\,.
\ee
As in the classical case, we are specially interested in the subalgebra of deformation operators which commute with the total area, $\hcE\equiv\sum_{i}\hcE^i=\f12\sum_{i}\beta_{ii}$. Everything happens exactly as at the classical level. Indeed we compute:
\be
[\hcE, \halpha^{z}_{ij}]=-i(\hbeta_{ij}-\hbeta_{ji})
\,,\qquad
[\hcE, \halpha^{w}_{ij}]=+i(\hbeta_{ij}-\hbeta_{ji})
\,,\qquad
[\hcE, \hbeta_{ij}]=-i(\halpha^{z}_{ij}-\halpha^{w}_{ij})
\,.
\ee
So we introduce the following linear combination of the deformation operators:
\be
\left|\begin{array}{l}
\hbeta^S_{ij}=\hbeta_{ij}+\hbeta_{ji}
\\
\hbeta^A_{ij}=\hbeta_{ij}-\hbeta_{ji}
\end{array}\right.
\,,\qquad
\left|\begin{array}{l}
\halpha^+_{ij}=\halpha^z_{ij}+\halpha^w_{ji}
\\
\halpha^+_{ij}=\halpha^z_{ij}-\halpha^w_{ji}
\end{array}\right.
\,,
\ee
so that the area-preserving deformation generators are:
\be
[\hcE,\halpha^+_{ij}]=[\hcE,\hbeta^S_{ij}]=0
\,,
\ee
and the area-changing deformation generators are:
\be
[\hcE,\halpha^-_{ij}]=-i\hbeta^A_{ij}\,,\qquad
[\hcE,\hbeta^A_{ij}]=-i\halpha^-_{ij}\,.
\ee
If we focus on the area-preserving operators, they form a $\gl_{N}(\R)$ algebra. It is actually convenient to repackage the $\halpha^+$'s and the $\hbeta^S$'s in a single matrix of Hermitian operators, $\hDelta_{ij}\equiv\halpha^+_{ij}+\hbeta^S_{ij}$. The commutators between the $\hDelta$'s are exactly the one of the canonical basis of the $\gl_{N}(\R)$ algebra.

This extends the $\u(N)$ structure of $\SU(2)$ intertwiners \cite{Girelli:2005ii,Freidel:2009ck,Freidel:2010tt,Livine:2013tsa} to the Lorentzian case and $\SU(1,1)$ intertwiners. 
Following the work done in the Euclidean case, the next steps to investigate would be:
\begin{itemize}

\item to decompose the space of $\SU(1,1)$-intertwiners in terms of irreducible $\GL_{N}(\R)$ representations:

Since the space of framed Lorentzian polyhedra with fixed total area $\cE$ carries a cyclic action of $\GL_{N}(\R)$, we expect that the Hilbert space of $\SU(1,1)$ intertwiners for a fixed eigenvalue of the total area $\hcE$ to be an irreducible representation of $\GL_{N}(\R)$. We need to prove this statement and identify the $\GL_{N}(\R)$-representations in terms of the value of $\cE$. Then we would decompose the whole of $\SU(1,1)$ intertwiners as a direct sum of those irreducible $\GL_{N}(\R)$-representations, with the operators $\halpha^+_{ij}$ and $\hbeta^S_{ij}$ generating the $\GL_{N}(\R)$ action within each irreducible representation while the operators $\halpha^-_{ij}$ and $\hbeta^A_{ij}$ would act as ladder operators going from one irreducible representation to another. This would define the fine structure of $\SU(1,1)$ intertwiners.

\item to define coherent states for $\SU(1,1)$ intertwiners:

The goal would be to identify  $\SU(1,1)$ intertwiners, at fixed total area, that transform coherently under the $\GL_{N}(\R)$ action. These would define semi-classical intertwiner states representing classical Lorentzian polyhedra. We could then glue those coherent intertwiners together into Lorentzian spin networks representing good semi-classical geometries for loop quantum gravity on time-like hypersurfaces. We could even go further, as in \cite{Girelli:2017dbk}, and identify intertwiners, with superpositions of the total area, that transform coherently under the whole group of deformations, with both area-preserving and area-changing transformations.

\end{itemize}

\section*{Outlook \& Conclusion}

We have defined the phase space for Lorentzian space-like vectors in the 3d Minkowski space $\R^{1,2}$ and use it to define the generalization of the Kapovitch-Millson phase space for flat Lorentzian polyhedra with space-like normals up to $\SO(2,1)$ Lorentz transformations. We have introduced a complete set of Lorentz-invariant observables, whose flow under the Poisson brackets generates deformations of the polyhedra. Distinguishing observables by the criteria of whether or not they Poisson-commute with the total boundary area of the polyhedron, this led us to identify $\GL_{N}(\R)$ as the group of area-preserving deformations of Lorentzian polyhedra with $N$ faces, thereby extending the work done with Euclidean polyhedra and $\U(N)$ deformations \cite{Freidel:2009ck,Freidel:2010tt,Livine:2013tsa,Girelli:2017dbk}. We further showed that this action is cyclic on the space of Lorentzian polyhedra at fixed total area and that one can generate any arbitrary polyhedra by acting with a $\GL_{N}(\R)$ transformations on a totally squashed polyhedric configuration with only two non-trivial faces. The $\GL_{N}(\R)$ transformations allow to blow up the polyhedron and give non-vanishing areas to all $N$ faces of the polyhedron. Finally, we showed how to quantize the vector phase space to recover all unitary representations of $\SU(1,1)$ from the principal continuous series thus identifying as quantum space-like vectors and how to quantize the polyhedron phase space to get $\SU(1,1)$ intertwiners identified as quantum Lorentzian polyhedra. Moreover the $\GL_{N}(\R)$ action is preserved at the quantum level and the Hilbert space of $\SU(1,1)$ intertwiners (at fixed total area) carries a (irreducible) representation of the $\GL_{N}(\R)$ group, which defines deformations of quantum polyhedra.

A next step would be to characterize the $\GL_{N}(\R)$ representations carried by the Hilbert space of $\SU(1,1)$ intertwiners and realize the space of quantum Lorentzian polyhedra with $N$ faces as a ladder of $\GL_{N}(\R)$ irreducible representations, with non-area-preserving deformations allowing to hop from one representation to another. The other question to investigate is the definition of coherent quantum polyhedra, most likely as $\GL_{N}(\R)$ coherent state \`a la Perelomov \cite{Perelomov:1986tf}, similarly to what has been done in the Euclidean case \cite{Freidel:2010tt,Girelli:2017dbk}.

Then we envision two fields of application of the present results.
At the one hand, in the context of discrete geometry, the $\GL_{N}(\R)$ transformations allow to explore the whole space of Lorentzian polyhedra and the $\GL_{N}(\R)$ Haar measure defines a probability measure for random Lorentzian polyhedra. Not only this can be an efficient technique if one seeks to generate and produce Lorentzian polyhedra with space-like normals, but one can also look for typicality results on Lorentzian polyhedra due to concentration of measure phenomena on $\GL_{N}(\R)$ at large number of faces $N$, similarly to what was studied in the Euclidean case \cite{Livine:2013tsa,Anza:2016fix,Anza:2017dkd}.
On the other hand, in the context of quantum gravity, our work finds a direct application in classifying and analyzing deformations of quantum time-like boundaries in the loop quantum gravity framework, which is relevant to the study of quasi-local observables and dynamics and of holographic dualities (for example to extend the approach of \cite{Dittrich:2018xuk} to  3+1d quantum gravity).

%

%
%


\bibliographystyle{bib-style}
\bibliography{SU11}

\end{document}